\documentclass[a4paper,11pt]{quantumarticle}
\pdfoutput=1
\usepackage[utf8]{inputenc}
\usepackage[english]{babel}
\usepackage[T1]{fontenc}
\usepackage{svg,times,upgreek}
\usepackage{color}
\usepackage{graphicx}
\usepackage{dcolumn}
\usepackage{bm}
\usepackage{physics}
\usepackage{hyperref}
\usepackage{url}
\usepackage[numbers,sort&compress]{natbib}

\usepackage{amsmath}
\usepackage{notes2bib}
\usepackage{ dsfont }
\usepackage{textcomp}
\usepackage{svg}
\usepackage{verbatim}
\usepackage{amssymb}
\usepackage{amsfonts}              
\usepackage{amsthm}                
\usepackage{mathtools}
\usepackage{kotex}
\usepackage[normalem]{ulem}
\usepackage{xcolor}
\usepackage{tikz}
\usetikzlibrary{quantikz}

\newcommand{\shorteq}{%
  \settowidth{\@tempdima}{-}
  \resizebox{\@tempdima}{\height}{=}%
}
\newcommand{\D}{\mathrm{d}}
\newcommand{\I}{\mathrm{i}}
\newcommand{\RE}[1]{\mathrm{Re}\{#1\}}
\newcommand{\IM}[1]{\mathrm{Im}\{#1\}}

\DeclarePairedDelimiterX{\barpair}[2]{(}{)}{%
  #1\;\delimsize\|\;#2%
}

\newtheorem{theorem}{Theorem}

\newcommand{\opt}{{p_\mathrm{hack}^\mathrm{opt}}}
\newcommand{\PG}{{p_\mathrm{hack}^\textsc{PG}}}
\newcommand{\ME}{{p_\mathrm{hack}^\textsc{ME}}}

\bibnotesetup{note-name    = ,use-sort-key = false} 
\begin{document}

\title{Repeated Extraction of Scrambled Quantum Data: Sustainability of the Hayden-Preskill Type Protocols}

\author{Seok Hyung Lie}
\author{Yong Siah Teo}
\author{Hyunseok Jeong}\email{h.jeong37@gmail.com}
\affiliation{%
 Department of Physics and Astronomy, Seoul National University, Seoul 08826, Republic of Korea
}%


\begin{abstract}
We {introduce and study} the problem of scrambler hacking, which is the procedure of quantum-information extraction from and installation on a quantum scrambler given only partial access. This problem {necessarily emerges from} a central topic in contemporary physics --- information recovery from systems undergoing scrambling dynamics, such as the Hayden--Preskill protocol in black hole studies {--- because one must replace quantum data with another when extracting it due to the no-cloning theorem.} {For large scramblers, we supply analytical formulas for the optimal hacking fidelity, a quantitative measure of the effectiveness of scrambler hacking with limited access}. In the two-user scenario where Bob attempts to hack Alice's data, we find that the optimal fidelity converges to $64/(9\pi^2)\approx0.72$ with increasing Bob's hacking space relative to Alice's user space. {We apply our results to the black hole information problem and show that the limited hacking fidelity implies the reflectivity decay of a black hole as an information mirror, which questions the solvability of the black hole information paradox through the Hayden-Preskill type protocol}. 
\end{abstract}

\maketitle

In many-body quantum systems, due to rapid and complex interaction between subsystems, initially localized quantum information dissipates quickly and spreads throughout the whole system. This delocalization of quantum information is called quantum scrambling, and retracting scrambled quantum information is one of the most important central topics of contemporary physics \cite{sekino2008fast,hosur2016chaos,blok2021scrambling}. From the perspective of the decoupling approach \cite{dupuis2010decoupling} to quantum information, perfect recovery is equivalent to implementing a quantum channel without leaking information to a particular subsystem, which is also studied in the context of catalysis of quantum randomness, quantum secret sharing and quantum masking \cite{boes2018catalytic, wilming2020entropy, lie2019unconditionally, lie2020randomness, lie2020uniform}.

One of the most important examples of recovery of quantum information from quantum scrambler is the Hayden--Preskill protocol for recovering quantum information from the Hawking radiation of old black holes \cite{page1994black, hayden2007black, bao2021hayden, cheng2020realizing, yoshida2019firewalls}. The facts that black hole evaporates by emitting  Hawking radiation which is predicted to be semi-classical thermal radiation and that no quantum information can be destroyed because of the unitarity of time evolution in quantum mechanics lead us to the famous black hole information paradox. The Hayden--Preskill protocol proposes a resolution for this paradox. The protocol is still under active research as its optimal decoding map is not completely understood and requires creative construction~\cite{yoshida2017efficient}.

Although the Hayden--Preskill protocol was proposed for extraction of quantum data from black holes, it can be applied to any quantum scrambler that allows attachment and detachment of subsystems. The setting of the protocol is as follows.  Suppose that Alice inputs a piece of quantum data into a part of the input register of a multipartite unitary operator (`scrambler') with a publicly known architecture. Bob acquires an access to a part of input/output ports of the network, and Bob attempts to extract as much data from Alice as possible. 

The Hayden--Preskill protocol says that, if Bob inputs a highly entangled state and the multipartite unitary operator has strong scrambling property, then by collecting a little more than $k$ qubits of output of the unitary operator one can successfully extract $k$ qubits of quantum data prepared by Alice \cite{hayden2007black}. Old black holes are believed to satisfy the aforementioned conditions, thus, old black holes `act like information mirrors'. This is the reason why the Hayden--Preskill protocol could possibly resolve the black hole information paradox.


 For the Hayden--Preskill protocol to be a true solution to the black hole information paradox, however, old black holes must function as information mirrors \textit{consistently}. In other words, one should be able to recover one qubit after another without limit, not just a few initial qubits. However, the protocol seemingly consumes entanglement in the process, thus one can naturally question if this protocol can be repeated for the extraction of the next qubits.
 
 Moreover, by the no-cloning theorem~\cite{wootters1982single}, Bob cannot simply copy out Alice's quantum data, but has to replace it with another. For the case of many-body systems such as black holes, the installed quantum data will be fed into the next round of scrambling as a part of input. Therefore, one might wonder if installing a certain piece of quantum data into quantum scramblers can enhance the performance of the data extraction for the next round, so that the \textit{sustainability} of quantum data extraction from quantum scramblers is better.
 
 Significance of quantum data installation in the studies on the Hayden--Preskill type protocols has been generally overlooked. For example, in the proposal of efficient decoding scheme for the Hayden--Preskill protocol by Yoshida and Kitaev \cite{yoshida2017efficient}, the remnant entangle state after decoding is only referred to as ``arbitrary state'' (in Figure 1 or Ref. \cite{yoshida2017efficient}) and no attention was given to its role in data extraction. In this work, we focus on this unnoticed, yet important problem of installing quantum data into the target system and how well this installation can be implemented at the same time with data extraction.

{\it Hayden--Preskill protocol.}--- The original Hayden--Preskill protocol can be stated as follows. An old black hole (denoted as $B$) is believed to be maximally entangled with all the Hawking radiation it has emitted (denoted as $B'$) so far. Now let $A$ be an infalling object. Under the assumption that black holes are undergoing rapid scrambling interaction  \cite{sekino2008fast}, we can say that systems $AB$ undergo a bipartite unitary operator $U:AB\to KL$. After this interaction, a part of output systems, $L$, whose size is slightly larger than $A$ is emitted from the black hole. By applying a suitable recovery map $R$ on $B'L$, one can recover the quantum state of $A$ almost perfectly \cite{hayden2007black}. One might interpret that quantum information is not destroyed in the process, so the black hole information paradox is resolved, at least from the perspective of the outer observer.

However, after this data extraction, the remaining black hole interior $K$ can interact with another infalling system $A_1$ and subsequently emit Hawking radiation $L_1$. What happens to the quantum information of $A_1$? Although systems $BB'$ were initially maximally entangled, the quantum correlation between the black hole and the data extractor may be degraded over time as the data extraction continues. The sheer fact that the size of $L$ must be larger than that of $A$ for accurate data recovery alone suggests that entanglement is consumed in the process. Is the Hayden--Preskill protocol sustainable after many rounds of data extraction?

The discussion so far can be applied to any quantum scrambler $U$ that behaves similarly. In the Hayden-Preskill setting, the initial maximal entanglement between a black hole and the systems outside is built through the natural Hawking radiation, but for general quantum scramblers, one may try to optimize data extraction by establishing a certain entangled state instead of the maximally entangled state. Also, as discussed before, one must install a certain quantum state in the target system when extracting quantum data because of the no-cloning theorem. Hence, one may want to leave some particular quantum state in the scrambler after data extraction. Would it be possible?

{\it Quantum-scrambler hacking}.--- Let us consider a general scrambler instead of a black hole. The quantum scrambler, accessed by only two users, Alice and Bob, is described as a $d_Ad_B$-dimensional unitary operator $U$ on systems $AB$.  After the interaction, the joint system is decomposed into $K$ and $L$, possessed by Alice and Bob, respectively. Note that it implies that $d_Ad_B=d_Kd_L$.

Alice inputs her secret quantum state in $A$, and Bob tries to extract as much quantum information stored in $A$ as possible and replace it with an arbitrary quantum state of his choice. This is done by feeding a `probe' quantum state that is going to interact with $A$ and applying a recovery map to the output state to simulate the \texttt{SWAP} operation \cite{garcia2013swap}.

However, if this task can be done with error $\epsilon$ (measured by the average infidelity of pure state inputs), then the unitary operator $U$ itself should be close to the \texttt{SWAP} operator (up to local operations) with error $\epsilon$ and vice versa. It is thus impossible to substitute quantum data through a non-\texttt{SWAP} operator~(See Appendix \ref{app:imp}).

\begin{figure}
	\centering
	\includegraphics[width=0.95\columnwidth]{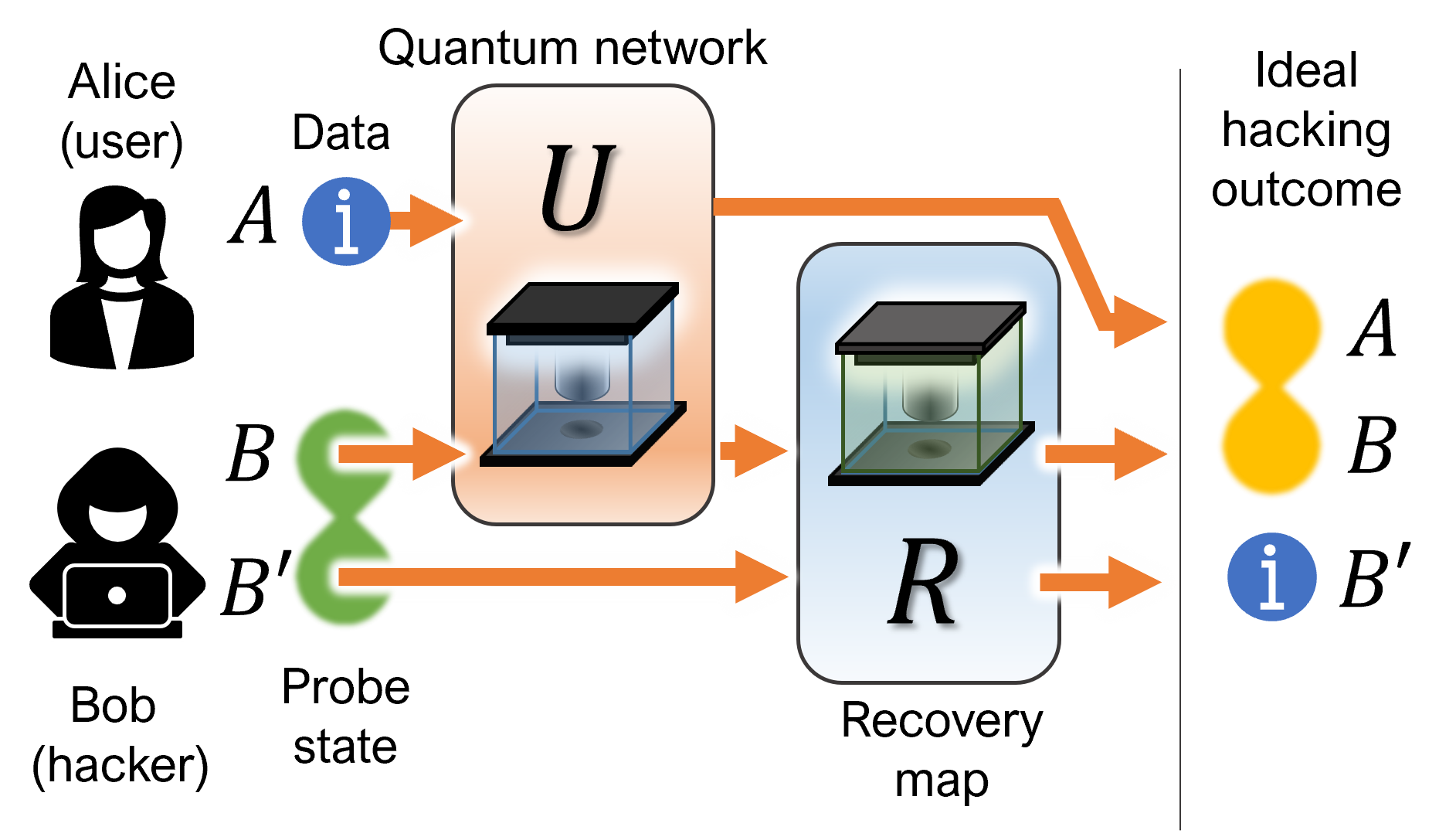}
	\caption{Schematic diagram of scrambler hacking of a unitary process $U$ (quantum scrambler or quantum computer) in the two-user scenario. Ideally, Bob, the hacker, would desire to extract Alice's information and plant a part of a maximally entangled state for the next quantum computation.}
	\label{fig:QHscheme}
\end{figure}

The next optimal strategy for Bob is to build as much correlation as possible with the target system $A$ while extracting quantum information out of it (See Fig.~\ref{fig:QHscheme}). This is because building correlations allows the extraction of quantum information from the next (undecided) computation step provided Bob has access to a part of the current computation output. In general, Bob possesses a reference system $B'$ ($d_B=d_{B'}$) and prepares an entangled input probe state $\ket{\phi}_{BB'}$. Bob's goals are, therefore, to extract the \emph{input} information stored in $A$, and install an \emph{output} maximally entangled state in $AB$. The latter can be interpreted as preparation of extraction of quantum data of \textit{future} quantum computation on $A$~\cite{horodecki2005partial, horodecki2005quantum}, because having a maximally entangled state with the target system yields the maximum side information.

Our protocol is set in the following situation.

$(i)$ A bipartite unitary operator $U:{AB\to KL}$ is randomly chosen, and Bob is informed about it.

$(ii)$ Alice prepares a quantum state in $A$ that Bob wants to extract. Bob prepares a quantum state $\ket{\phi}_{BB'}$ (`probe state') with an ancillary system $B'$.

$(iii)$ Unitary operator $U$ acts on $AB$. Alice retains system $K$ and Bob retains system $L$.

$(iv)$ Bob applies a recovery unitary operator $R$ on $LB'$. The output system decomposes into $A''$, where the extracted quantum state of $A$ should be prepared at, and the rest which should be correlated with $K$.

$(v)$ Another quantum state is prepared in $A_1$, and another random bipartite unitary operator $U_1$ on $A_1B_1\to K_1L_1$ where $B_1:=K$ is chosen and known to Bob. After the action of $U_1$, Alice retains system $K_1$ and Bob acquires system $L_1$. Bob applies a recovery map $R_1$ on $KL_1$ and Bob repeats similar steps.

Steps $(i)$ to $(iv)$ form the first round of task, and as step $(v)$ explains, similar rounds are repeated except that the probe state is not prepared by Bob, but is automatically prepared from the previous round. We will call the task that tries to achieve the both goals \textit{quantum-scrambler hacking} because of its resemblance to conventional data hacking in which a hacker extracts a piece of data and install another on the remote system.

To achieve both goals, Bob applies a unitary recovery map $R$ on systems $BB'$. Since the input data in $A$ is unknown, for the purpose of evaluating Bob's strategy following the standard approach \cite{hayden2007black, yoshida2017efficient}, we assume that it is in a maximally entangled $\ket{\psi}_{AA'}$ state with an environment $A'$ ($d_A=d_{A'}$), where $\ket{\psi}_{XY}=\sum_{i=1}^{d'}\ket{ii}_{XY}/\sqrt{d'}$ 
and $d'=\min\{d_X,d_Y\}$ for any systems $XY$. For this case, successful data extraction from $A$ means entanglement swapping; Bob should have $\ket{\psi}_{A'B'}$ in the final stage. The fidelity with maximally entangled inputs is known to be a monotone function of the average fidelity of a pure-state input \cite{horodecki1999general}. Consequently, we define the scrambler hacking fidelity as
\begin{equation} 
    \label{eqn:phack}
    p_{\mathrm{hack}}:=|\bra{\psi}_{AB'}\bra{\psi}_{A'A_1}R_{LB'}U_{AB}\ket{\psi}_{AA'}\ket{\phi}_{BB'}|^2.
\end{equation}

Since the fidelity never decreases under a partial trace, $p_\mathrm{hack}$ serves as a lower bound for both fidelities of the extracted quantum data (systems $A'B'$) and implemented entangled state (systems $AB$). We can parametrize any bipartite entangled pure state $\ket{\phi}_{BB'}$ with an operator $\chi$ acting on system $B'$ such that $\ket{\phi}_{BB'}=\sum_i \ket{i}_B\otimes \chi\!\ket{i}_{B'}$  with $\|\chi\|_2=1$. Here $\|X\|_p:=(\Tr|X|^p)^{1/p}$, where $|X|:=\sqrt{X^\dag X}$, is the Schatten $p$-norm. With these, Eq.~(\ref{eqn:phack}) is simplified to
\begin{equation} \label{eqn:PVhack}
    p^{(R,\chi)}_{\mathrm{hack}}=|\Tr[R(I_L\otimes \chi)U^o]|^2/(d_A^2 d_K).
\end{equation}
Here, the (generally non-unitary) map $U^o : AA' \to BB'$ is represented by a matrix, understood as a tensor, formed by cyclically rotating the indices of $U$ clockwise by one position---${U^o}^{ij}_{kl}:=U^{ki}_{lj}$. Here, $X^{ij}_{kl}:=\bra{ij}X\ket{kl}$ in the computational basis~\cite{SM_Qhack}. This amounts to {a clockwise \mbox{($\pi/2$)-rotation} of $U$} in tensor-network diagrams~\cite{montangero2018introduction, landsberg2011geometry}, and is closely related to tensor reshuffling~\cite{zyczkowski2004duality, miszczak2011singular, bruzda2009random}. Note that $\|U^o\|_2=\|U\|_2=\sqrt{d_Ad_B}$~\bibnote{Although $R$ is $d_B^2\times d_B^2$, it only acts on a $d_A^2$-dimensional subspace $\Im \{(I_B \otimes \chi)U^o\}$ to the right and $(\mathrm{Ker } \{U^o\})^\perp$ to the left, so that we may treat $R$ either as a rank-$d_A^2$ partial unitary matrix or a $d_A^2 \times d_B^2$ coisometry without losing generality.}.

\begin{figure}
	\centering
	\includegraphics[width=0.95\columnwidth]{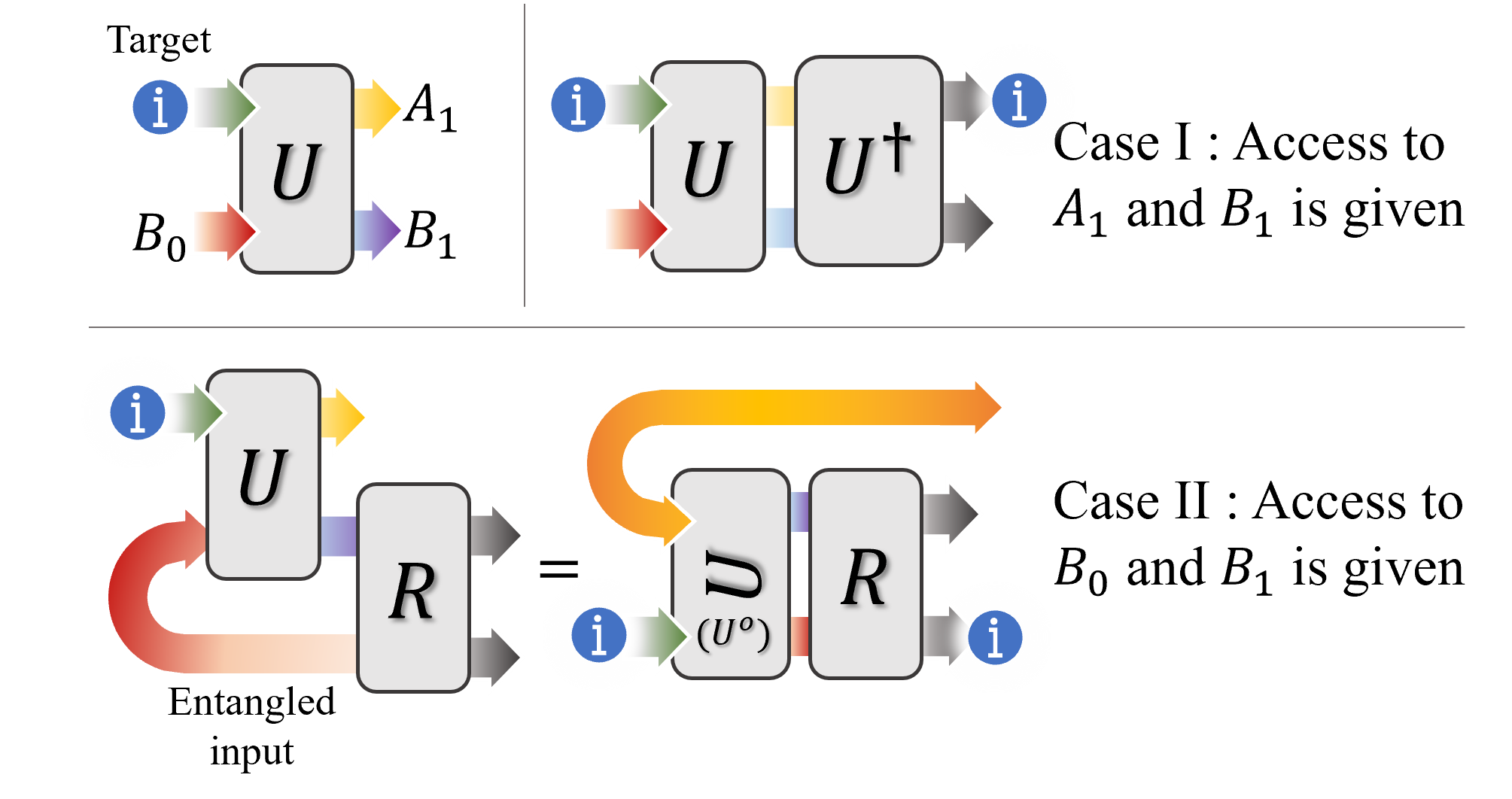}
	\caption{scrambler hacking and {the} operational meaning of $U^o$. {When all output ports of the scrambler $U$ are fully accessible to the attacker~(Case~I), data extraction is as simple as applying the complete recovery map $U^\dag$. This is in contrast to the more realistic situation of scrambler hacking~(Case~II) in which the hacker only has \emph{partial access} to the input and corresponding output ports of $U$. In this case, hacking amounts to inverting a possibly non-unitary operator---the ($\pi/2$)-rotated $U$ or~$U^o$---with a unitary recovery map $R$. Therefore, the degree of unitarity of $U^o$ is directly related to the performance of scrambler hacking.}}
	\label{fig:Qhack}
\end{figure}

Each pair $(R,\chi)$ constitutes a hacking strategy for Bob. For a given $\chi$, which is identical to fixing Bob's probe state, the optimal unitary recovery~$R$ is the one that gives the polar decomposition $(I_L\otimes \chi)U^o=R^\dag|(I_L\otimes \chi)U^o|$. This leads to
\begin{equation}
    p^{(\chi)}_{\mathrm{hack}}=\max_R\, p^{(R,\chi)}_{\mathrm{hack}}=\|(I_L\otimes \chi)U^o\|_1^2/{(d_A^2 d_K)}\,,
    \label{eqn:phack_optR}
\end{equation}
which is equivalent to inverting a possibly non-unitary operator with a unitary one~\cite{lesovik2019arrow}.

This leaves us the problem of finding an optimal $\chi$ that achieves the largest hacking fidelity. A natural candidate would be $\chi=I_B/\sqrt{d_B}$, which corresponds to a maximally entangled probe state. Equation~(\ref{eqn:phack_optR}) then immediately yields the fidelity $p_{\mathrm{hack}}^\textsc{me}=\|U^o\|_1^2/{(d_A^{2}d_Bd_K)}$, which is a \emph{unitarity measure} of $U^o$ as $\|U^o\|_1$ is the maximal inner product of $U^o$ and an isometry. With this strategy, perfect hacking ($p_{\mathrm{hack}}^\textsc{me}=1$) is only possible when $U^o$ is proportional to an isometry---$U^{o\dag} U^o={(d_B/d_K)I_{AA'}}$. {Unitary operators $U$ with such a property are known to be \textit{dual-unitary} in the studies of quantum lattice models \cite{kos2021correlations, bertini2019exact,aravinda2021dual}, and we give a new operational meaning to them as \textit{completely hackable} unitary operators.} On the contrary, $p_{\mathrm{hack}}^\textsc{me}$ reaches its minimum $1/d_A^2$ when $U^o$ is rank-1, which happens if $U=I_A\otimes I_B$, for instance. This value serves as a lower bound of the optimal hacking fidelity {and shows that the degree of unitarity of $U^o$ directly affects the performance of scrambler hacking. (See FIG. \ref{fig:Qhack}.)}

Physical intuition may lead to the putatively obvious conclusion that a maximally entangled probe state is optimal for scrambler hacking. As it turns out, this is, however, not the case in general. For example, for a qubit-qudit controlled unitary operator given as $U_c=I_A\otimes\dyad{0}_B+X_A\otimes(I_B-\dyad{0}_B)$, with the Pauli $X$ operator acting on $A$, $p_{\mathrm{hack}}^\textsc{me}$ is smaller than the $p^{(\chi)}_{\mathrm{hack}}$ with $\chi=(\dyad{0}_B+\dyad{1}_B)/\sqrt{2}$. 

To maximize $p^{(R,\chi)}_\mathrm{hack}$ in Eq.~(\ref{eqn:PVhack}), recalling that $\|\chi\|_2=1$, we may invoke the Cauchy--Schwarz inequality, $|\Tr_{B'}[\chi \Tr_B[U^oR]]|\leq \|\Tr_B[U^oR]\|_2$. This bound is saturated when $\chi=\Tr_B[R^\dag U^{o\dag}]/\|\Tr_B[U^oR]\|_2$. Hence, the true optimal hacking fidelity reads
\begin{equation} \label{eqn:opthack}
    p_{\mathrm{hack}}^{\mathrm{opt}}=\max_R\|\Tr_B[U^o R]\|_2^2/d_A^3,
\end{equation}
where the maximization is over all $d_A^2 \times d_B^2$ coisometry operator $R$~($RR^\dag=I_{AA'}$). By exploiting the polar decomposition once more, a natural choice of $R$ is $U^oR=|U^{o\dag}|$ and yields the fidelity $p^{\mathrm{PG}}_{\mathrm{hack}}={\|\Tr_B|U^{o\dag}|\|_2^2}/{d_A^3}$. As we shall soon demonstrate that this hacking strategy is near-optimal, we will call this the ``pretty good'' (PG) strategy. As this strategy also outperforms that using a maximally entangled probe state, the following inequalities hold:
\begin{equation} \label{eqn:bds}
    p^{\textsc{ME}}_{\mathrm{hack}}\leq p^{\mathrm{PG}}_{\mathrm{hack}} \leq p^\mathrm{opt}_{\mathrm{hack}}\,.
\end{equation}

Note that both ME and PG strategies are \emph{optimal} when $U^o$ is proportional to an isometry $(p^{\mathrm{ME}}_{\mathrm{hack}}=p_\mathrm{hack}^\mathrm{PG}=p_\mathrm{hack}^\mathrm{opt}=1)$. Conversely, if $p_\mathrm{hack}^\mathrm{opt}\approx 1$, then the output state $R_{BB'}U_{AB}\ket{\psi}_{AA'}\ket{\phi}_{BB'}$ is almost maximally entangled with respect to the partition $AA'|BB'$, i.e., $S(AA')\approx 2\log_2 d_A$. Thus, the von Neumann entropy of system $B'$ is lower bounded as $S(B')\geq S(AA')-S(B)\gtrsim 2\log_2 d_A - \log_2 d_B$. Considering that $S(B')\leq \log_2 d_B$, it follows that nearly perfect hacking is possible only when $d_B \gtrsim d_A$. Moreover, if $d_B\approx d_A$, then the probe state is forced to be near-maximally entangled, i.e. $S(B')=S(|\chi_{\mathrm{opt}}|^2)\approx \log_2 d_B$.

Especially, If $d_A=d_B$, the inequality ${1-p_\mathrm{hack}^\textsc{ME}}\leq 4({1-p_\mathrm{hack}^\mathrm{opt}})$ holds and it implies that near-perfect hacking $(p_\mathrm{hack}^\mathrm{opt} \approx 1)$ is possible only when $U^o$ is nearly unitary $(p_\mathrm{hack}^\textsc{ME} \approx 1)$~(See Appendix~\ref{app:dual}). {This suggests that properties of the ostensibly abstract $(\pi/2)$-rotated unitary operator $U^o$, not only those of $U$ itself, directly reflect the underlying operational meaning of scrambler hacking, namely the `inversion of a bipartite unitary operator from sideways in tensor network', as depicted in Fig.~\ref{fig:Qhack}.}

We remark that the apparently harder problem of finding optimal strategy of scrambler hacking, which aims for quantum data installation in addition to data extraction, is actually computationally equivalent to finding an optimal choice of decoding of the Hayden--Preskill-type protocols. Hence, scrambler hacking costs no more than usual quantum data extraction in both optimization and physical implementation. See Appendix~\ref{app:two-qubit} for the detailed discussion.

\begin{figure}[t]
	\centering\includegraphics[width=0.95\columnwidth]{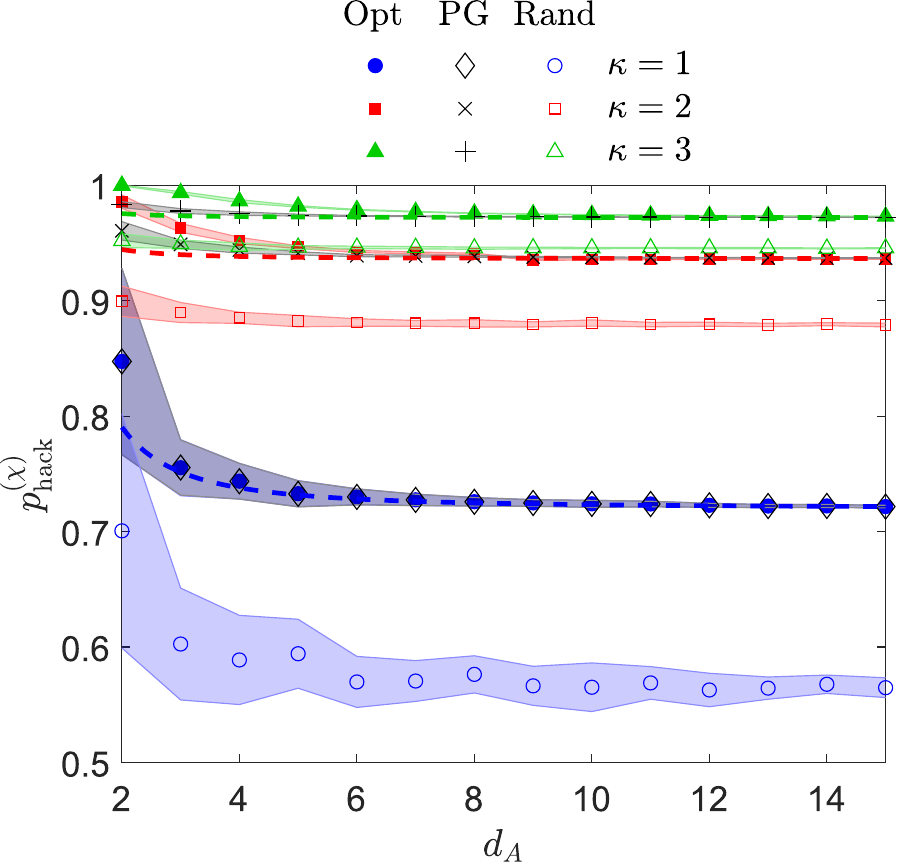}
	\caption{\label{fig:avg_phackopt}{Averaged quantum-scrambler hacking performance (over 20 randomly-generated Haar unitary scramblers {$U$}) featuring the optimal strategy~(Opt) \emph{via}~\eqref{eqn:opthack}, the PG strategy with $\widetilde{\chi}$, and a random one~(Rand) using an arbitrarily-chosen probe state. Without loss of generality, we show graphs only for $\kappa\geq1$. When $\kappa=1$ (corresponding to the Hayden--Preskill scenario), PG is almost the same as Opt in hacking performance. As $\kappa$ increases,  $p^\mathrm{opt}_\mathrm{hack}\rightarrow\mathcal{I}_{{\kappa\gg1}}^2\rightarrow1$. All theoretical dashed curves are computed with~\eqref{eqn:asymp}.}}
\end{figure}

{\it Optimal scrambler-hacking performance}.---For a quantum scrambler described by a generic unitary $U$ and an optimal hacking strategy defined by optimizing $R$ and $\chi$, the hacking fidelity $p_\mathrm{hack}$ defined in Eq.~\eqref{eqn:phack} reaches an optimal value $p^\mathrm{opt}_\mathrm{hack}$ stated in Eq.~\eqref{eqn:opthack}. As previously discussed, the optimal probe state~($\chi_\mathrm{opt}$) for $p^\mathrm{opt}_\mathrm{hack}$ is generally not maximally entangled ($\chi_\mathrm{opt}\neq I_B/\sqrt{d_B}$). While the general solution of $\chi_\mathrm{opt}$ for the maximization problem in Eq.~\eqref{eqn:opthack} has no known analytical form, {we propose two numerical methods to acquire $\chi_\mathrm{opt}$ in Appendix~\ref{app:numer}: a repeated iteration of the corresponding extremal equation, and a gradient-ascent algorithm~\cite{teo2011mlme,teo2011adaptive}.} 

Interestingly, one can calculate an analytical form of $p^\mathrm{opt}_\mathrm{hack}$ for sufficiently large dimensions{; more specifically---$d_Ad_K\gg1$.} To do this, we observe that a maximally entangled probe state is asymptotically optimal for {quantum-scrambler hacking}---$\chi_\mathrm{opt}\rightarrow I_B/\sqrt{d_B}$. This follows from the fact that the reduced state of any high-dimensional pure state approaches the maximally mixed state~\cite{page1993entropy}, which then implies that $p^\mathrm{opt}_\mathrm{hack}\rightarrow\| U^o\|_1^2/{(d_A^2d_Bd_K)}$. We shall consider $U$ as a random unitary operator distributed according to the Haar measure of the unitary group. Using properties of this measure and results from random matrix theory~\cite{mehta2004random,marcenko1967eigenvalues,SpFuncBk}, {in the scenario where only Alice and Bob are influenced by the action of} a generic quantum scrambler{, we have the asymptotic Haar-averaged formula for $\kappa\equiv \sqrt{d_Bd_L/(d_Ad_K)}=d_B/d_K\geq1$,}
\begin{gather}
	\overline{p^{\mathrm{opt}}_\mathrm{hack}}\approx\mathcal{I}_\kappa^2+(1-\mathcal{I}_\kappa^2)/{(d_Ad_K)}\,,\nonumber\\
	\,\,\mathcal{I}_\kappa={}_2\mathrm{F}_1\left({2^{-1},-2^{-1};2;\kappa^{-2}}\right)\,.
	\label{eqn:asymp}
\end{gather}
with {${}_2\mathrm{F}_1(\,\cdot\,,\,\cdot\,;\,\cdot\,;\,\cdot\,)$} being the hypergeometric function~\cite{SM_Qhack}. {In the specific circumstance where $U$ is the interaction unitary operator for the Hayden--Preskill scenario, we have $\kappa=1$, so that} $\mathcal{I}_1=8/(3\pi)$ or $\overline{p^{\mathrm{opt}}_\mathrm{hack}}\approx 0.72$. If $\kappa<1$, we instead have $\overline{p^{\mathrm{opt}}_\mathrm{hack}}\approx{\kappa^2\mathcal{I}_{1/\kappa}^2+(1-\mathcal{I}_{1/\kappa}^2)/(d_Ad_K)}$. 

Figure~\ref{fig:avg_phackopt} shows the performances of three hacking strategies with optimal recovery $R$ [see \eqref{eqn:phack_optR}]. The results indicate that efforts in using optimal probe states for a given scrambler $U$ do pay off with a much higher hacking fidelity compared to all other random choices of Bob's probe state. The quantity $\mathcal{I}_\kappa^2$ is an important indicator of the limiting performance for hacking large quantum scramblers of a fixed dimension ratio $\kappa$. It also suggests that in the two-user scenario, a larger Hilbert space of Bob relative to Alice's results in a larger $p^{\mathrm{opt}}_\mathrm{hack}$. A single-qubit ancilla ($\kappa=2$) is enough to boost $p^{\mathrm{opt}}_\mathrm{hack}$ all the way to $\approx0.936$. We also find that the PG strategy $\widetilde{\chi}:=\chi={\Tr_L|U^{o\dag}|/\|\Tr_L|U^{o\dag}|\|_2}$ is almost optimal for any $d_A$ and $d_B$. In particular, when $d_A=2=d_B$, we precisely get $\widetilde{\chi}=\chi_\mathrm{opt}=I_B/\sqrt{d_B}$~{(see Appendix~\ref{app:two-qubit})}.

{\it Scrambler hacking and entanglement recycling}.---The operator $U^o$ appears in various scenarios. In some, the input and output systems of the unitary operator $U$ need not match. Let $U:AB\to KL$ be a map of dimension $D=d_Ad_B=d_Kd_L$. By using a probe state ($\chi$) and recovery map $R$ on $LB'$ with matching dimensions, we get the modified hacking fidelity $p_\mathrm{hack}^{(R,\chi)}={|\Tr[R(I_L\otimes \chi)U^o]|^2}/{(d_A^2 d_K)}$. {One remarkable case is where Bob receives continual emission of data packets from the target system.

An example of such a situation is data extraction from the} evaporation of an old black hole, where the probe state is maximally entangled between the inner degrees of freedom of the black hole and all Hawking radiation emitted from the back hole up to that point. The black hole degrees of freedom is typically much larger than those of matter falling into it momentarily. Let the former be $D_B=d_B=d_K$ and the latter be, say, qudit: $d_M=d_A=d_L$. If Bob collects an additional $d_M$-dimensional Hawking radiation, the resulting optimal hacking fidelity is
$p_\mathrm{hack}^\mathrm{BH}=\|U^o\|_1^2/D^2$.

We remark that the dimension of the black hole interior state remains the same since a qudit enters the black hole and another exits it. Depending on the assumptions made on the dynamics of black holes (see \cite{tajima2021symmetry} for the recent discussion on the effect of symmetry for information recovery), there may be an estimated value $1-\epsilon_{\mathrm{BH}}$ of $p_\mathrm{hack}^\mathrm{BH}$. For example, for Haar random $U$, $p_\text{hack}^\text{BH}$ tends to $(8/3\pi)^2\approx 0.72$ for large $D$~\cite{SM_Qhack}. This also serves as a lower bound for the fidelity between the posterior probe state and a maximally entangled state. If one uses a probe state whose maximal fidelity with a maximally entangled state is $f$ for the information extraction of the next qudit falling into the black hole, where optimal hacking fidelity is typically $p_\mathrm{hack}^\mathrm{BH}$ for large $D$, then the optimal hacking fidelity approaches to the product $fp_\mathrm{hack}^\mathrm{BH}$. However, for a given hacking fidelity $p_\mathrm{hack}$ and the fidelities of the extracted quantum data ($f_\mathrm{ext}$), and between the posterior probe state and a maximally entangled state~($f_\mathrm{post}$), the following trade-off relation exists~\cite{SM_Qhack}:
\begin{equation}\label{eqn:fidtrade}
    f_\mathrm{ext} + f_\mathrm{post} \leq 1+ p_\mathrm{hack}.
\end{equation}
So one should choose between accurate data extraction and good entanglement recycling for imperfect hacking $(p_\mathrm{hack}<1)$; giving up the former means inaccurate data extraction for the current round of hacking, and giving up the latter leads to a worse fidelity in the next round. {We remark that the same argument can be applied to any large and generic quantum scrambler, not only to black holes, hence the same issue of entanglement recycling happens there, too. It suggests that if the target quantum scrambler does not allow perfect scrambler hacking, then the quality of the extracted data must be degraded over many rounds extraction.}

\textit{Discussion.---}We proposed a scrambler hacking task, which entails the extraction and replacement of quantum data through limited interaction with a quantum scrambler, and analyzed its performance in terms of the hacking fidelity. In finding good hacking strategies and calculating the optimal hacking fidelity for generic multipartite unitary scramblers, we give explicit operational meaning to {$\pi/2$} tensor rotations of unitary operators. {Moreover, we proved that finding an optimal decoder for this stronger task is equivalent to that for Hayden--Preskill-type protocols.}

While Bob's hacking fidelity on Alice saturates at a nonzero value in the two-network-user scenario, we find that with multiple users, naive attempts to hack any single user would generally result in exceedingly-low hacking fidelity. To improve the hacking success, it is necessary to perform program modifications to $U$, akin to hackers introducing malware to control classical computers. For quantum networks, quantum circuits would constitute such a program, but since an arbitrary quantum program cannot be encoded into a state owing to the no-programming theorem~\cite{nielsen1997program}, Bob would need to supplement his quantum resources with additional classical attacks to improve the hacking fidelity.

As an interesting application, we considered an information-reflection model for black holes, and surveyed the sustainability of black hole mirroring. From our trade-off relation in~(\ref{eqn:fidtrade}) and analysis of black hole hacking, we conclude that the black hole in this model indeed functions as a mirror \cite{hayden2007black}, but its ``reflectivity'' may be gradually degraded over time (See Ref. \cite{oshita2020reflectivity, wang2020echoes} for different notions of reflectivity of quantum black holes). One could collect more Hawking radiation to increase the hacking fidelity, but that necessitates an entanglement reduction~\cite{hayden2007black}, thereby leading to yet another type of reflectivity degradation. This questions whether quantum information gets destroyed when it falls into an `older' black hole that has already reflected a significant amount of quantum information that fell into it. 

\begin{acknowledgments}
This work is supported by the National Research Foundation of Korea (NRF-2019R1A6A1A10073437, NRF-2019M3E4A1080074, NRF-2020R1A2C1008609, NRF-2020K2A9A1A06102946) via the Institute of Applied Physics at Seoul National University and by the Ministry of Science and ICT, Korea, under the ITRC (Information Technology Research Center) support program (IITP-2020-0-01606) supervised by the IITP (Institute of Information \& Communications Technology Planning \& Evaluation).
\onecolumn \newpage

\appendix

\section{Implausibility of arbitrary quantum-state installation}
\label{app:imp}
For any $d$-dimensional quantum channel $\Lambda$, the infidelity $1-\mathcal{F}_{\mathrm{ent}}(\Lambda)=1-\bra{\psi}(\mathcal{I}\otimes\Lambda)(\dyad{\psi}))\ket{\psi}$ for maximally entangled input state $\ket{\psi}=\sum_{i=1}^d\ket{ii}$ and the average infidelity over Haar random pure input state $1-\mathcal{F}_{\mathrm{pure}}(\Lambda)=1-\int d\phi \bra{\phi}\Lambda(\dyad{\phi})\ket{\phi}$ have the following linear dependence \cite{horodecki1999general},
\begin{equation}1-F_\mathrm{ent}(\Lambda)=\frac{d}{d+1}(1-\mathcal{F}_\mathrm{pure}(\Lambda))\,,
\end{equation}
we will use the infidelity $1-F_\mathrm{ent}(\Lambda)$ instead and have the equivalent result without losing generality.

Suppose, with the unitary operators $U$, $P$ and $R$, that the \texttt{SWAP} operator can be approximated with error $\epsilon$ according to

\begin{equation}
\begin{quantikz}[column sep=0.3cm]
\lstick{$A$} & \qw   &\gate[2]{U}&\qw&\qw          \qw\\
\push{\ket{0}_B} & \gate[2]{P}   &    &\gate[2]{R}&\meterD{\Tr}\\
\lstick{$C$} &   &\qw    & &\qw                    \qw
\end{quantikz} \approx_\epsilon \begin{quantikz} [column sep=0.3cm]
\lstick{A}&\gate[swap, style={draw=white}]{}&\qw\\
\lstick{C}&&\qw
\end{quantikz}.
\end{equation}
Here, $\approx_\epsilon$ means that the two circuits are close to each other with error $\epsilon$ in the infidelity for maximally entangled input states. It means that $\Sigma \approx_\epsilon \Lambda$ is equivalent to $F(J_\Sigma,J_\Lambda)\geq 1-\epsilon$ where $J_\mathcal{N}$ is the normalized Choi matrix for quantum channel $\mathcal{N}$. From the Uhlmann theorem \cite{uhlmann1976transition}, it follows that there exists a pure state $\ket{s}_B$ such that
\begin{equation}
\begin{quantikz}[column sep=0.3cm]
\lstick{$A$} & \qw   &\gate[2]{U}&\qw&\qw          \qw\\
\push{\ket{0}_B} & \gate[2]{P}   &    &\gate[2]{R}&\qw\\
\lstick{$C$} &   &\qw    & &\qw                    \qw
\end{quantikz} \approx_{\epsilon} \begin{quantikz} [column sep=0.3cm]
\lstick{$A$}&\gate[swap, style={draw=white}]{}&\qw\\
\lstick{$C$}&&\qw\\
\lstick{$\ket{s}_B$}&\qw&\qw
\end{quantikz}.
\end{equation}
Since the fidelity never decreases under partial trace, it follows that
\begin{equation}
\begin{quantikz}[column sep=0.3cm]
\lstick{$A$}&\qw&\gate[2]{U}&\qw\\
\lstick{$C$}&\gate{\Phi}&&\qw
\end{quantikz} \approx_{\epsilon} { \begin{quantikz} [row sep=0.4cm, column sep=0.3cm]
	\lstick{$A$}&\gate[swap, style={draw=white}]{}&\qw&\qw\\
	\lstick{$C$}&&\gate{\Psi}&\qw
	\end{quantikz}}
\end{equation}
where
\begin{equation}
\begin{quantikz}[column sep=0.3cm] \lstick{$C$}&\gate{\Phi}&\qw\rstick{$B$}\end{quantikz}=
\begin{quantikz}[column sep=0.3cm]
\lstick{$\ket{0}_B$}&\gate[2]{P}&\qw\\
\lstick{C}&&\meterD{\Tr}
\end{quantikz},
\end{equation}
and
\begin{equation}
\begin{quantikz}[column sep=0.3cm] \lstick{$C$}&\gate{\Psi}&\qw\end{quantikz}=
\begin{quantikz}[column sep=0.3cm]
\lstick{$\ket{s}_B$}&\gate[2]{R^\dag}&\meterD{\Tr}\\
\lstick{$C$}&&\qw
\end{quantikz}.
\end{equation}
From the cyclicity of the fidelity for the maximally entangled input state, it follows that
\begin{equation}
\begin{quantikz}[column sep=0.3cm]
\lstick{$A$}&\gate[2]{U}&\qw\\
\lstick{$B$}&&\qw
\end{quantikz} \approx_{\epsilon} \begin{quantikz} [column sep=0.3cm]
\lstick{A}&\gate[swap, style={draw=white}]{}&\gate{\Phi^\dag}&\qw\\
\lstick{B}&&\gate{\Psi}&\qw
\end{quantikz}.
\end{equation}

It can be interpreted that unless the target scrambler $U$ itself is already close to a swapping operator followed by local operations, it is impossible to nearly perfectly substitute the quantum information out of the scrambler. Conversely, if $U$ is close to the \texttt{SWAP} operator with error $\epsilon$ (if the dimensions of $A$ and $B$ do not match, then it can be $\texttt{SWAP} \oplus I$), then by choosing $P$ and $R$ also as \texttt{SWAP} operators, one can achieve data substitution with error $\epsilon$.

\section{Rotation of matrix and fidelity bounds}
\label{app:rotation}

We first show how the fidelity expression
\begin{equation}
p_{\mathrm{hack}}^{(R,\chi)}:=|\bra{\psi}_{AB}\bra{\psi}_{A'B'}R_{BB'}U_{AB}\ket{\psi}_{AA'}\ket{\phi}_{BB'}|^2\;,
\end{equation}
is simplified with the rotated matrix $U^o$. First, note that $p_{\mathrm{hack}}$ is the fidelity between the pure states
\begin{equation}
\begin{quantikz}
&\makeebit[0]{$\dfrac{1}{\sqrt{d_A}}$}&\qw&\qw&\qw\rstick{$A'$} \\
&&\gate[2]{U}&\qw&\qw\rstick{$A$}\\
&\makeebit[0]{}&\qw&\gate[2]{R}&\qw\rstick{$B$}\\
&&\gate{\chi}&&\qw\rstick{$B'$}
\end{quantikz}\text{ and }
\begin{quantikz}[row sep=0.75cm]
&\makeebit[0]{}&\qw&\qw\rstick{$A'$} \\
&\makeebit[0]{$\dfrac{1}{d_A}$}&\makeebit[0]{}&\qw\rstick{$A$}\\
&\makeebit[0]{}&&\qw\rstick{$B$}\\
&&\qw&\qw\rstick{$B'$}
\end{quantikz}.
\end{equation}
Here, $\!\!\!\!\!\!\!\!\begin{quantikz}[row sep=0.3cm]&\makeebit[0]{}&\qw\\&&\qw\end{quantikz}=\sum_i \ket{ii}$ represents an unnormalized maximally entangled state with the appropriate Schmidt number for the system it is defined on. Therefore, it can be expressed with a tensor network diagram.
\begin{equation} \label{eqn:phackdiag}
p_{\mathrm{hack}}^{(R,\chi)}=\frac{1}{d_A^{3}}\!\!\!\!\!
\begin{quantikz}
&\makeebit[0]{}&\qw&\qw&\qw&\qw&\qw\makeebit[0]{} \\
&&\gate[2]{U}&\qw&\qw\makeebit[0]{}&&\makeebit[0]{}\\
&\makeebit[0]{}&\qw&\gate[2]{R}&\qw&&\makeebit[0]{}\\
&&\gate{\chi}&&\qw&\qw&\qw
\end{quantikz}
\;,
\end{equation}
with the equivalence of $\ket{\phi}_{BB'}$ and $\chi$:
\begin{equation}
\ket{\phi}_{BB'}=\begin{quantikz}
\makeebit[0]{}&\qw&\qw\\
&\gate{\chi}&\qw
\end{quantikz}\;.
\end{equation}
We remark that the time flows from left to right in the diagram, as opposed to the matrix multiplication order. The definition of $U^o$ can be expressed in a circuit diagram as follows:
\begin{equation}
\begin{quantikz}
&\gate[2]{U^o}&\qw\\
&&\qw\\
\end{quantikz}=
\begin{quantikz}
&\qw&\qw\makeebit[0]{}\\
&\gate[2]{U}&\qw\\
\makeebit[0]{}&&\qw\\
&\qw&\qw
\end{quantikz}
\;.
\end{equation}
By plugging this diagram into Eq. (\ref{eqn:phackdiag}), we get
\begin{equation}
p_{\mathrm{hack}}^{(R,\chi)}=\frac{1}{d_A^{3}}\!\!\!\!\!
\begin{quantikz}
&\makeebit[0]{}&\qw&\qw&\qw&\qw\makeebit[0]{}\\
&&\gate[2]{U^o}&\qw&\gate[2]{R}&\qw\\
&\makeebit[0]{}&&\gate{\chi}&&\qw\makeebit[0]{}\\
&&\qw&\qw&\qw&\qw
\end{quantikz}\;.
\end{equation}
This is equivalent to the expression
\begin{equation}
p_\mathrm{hack}^{(R,\chi)}=\frac{|\Tr[R(I_B \otimes \chi)U^o]|^2}{d_A^3}.
\end{equation}

Since $\sqrt{d_B}\|\Tr_B|U^{o\dag}|\|_2\geq\|\Tr_B|U^{o\dag}|\|_1=\|U^{o}\|_1$, $p_\mathrm{hack}^\textsc{PG}=\|\Tr_B|U^{o\dag}|\|_2^2/d_A^3$ is higher than $p^{\mathrm{ME}}_\mathrm{hack}=\|U^o\|_1^2/(d_A^3d_B).$ Also, since the PG strategy is a particular strategy, the fidelity of it is not larger than that of the optimal strategy, so we have $\PG \leq \opt$. In summary, we have

\begin{equation}
\ME \leq \PG \leq \opt.
\end{equation}

When $d_A=d_B=d,$ $\opt$ is the maximal fidelity between a maximally entangled state with the Schmidt rank $d^2$ and a pure state of the form $\Omega_\chi:=d(U_{AB}\otimes\chi_{B'})\dyad{\psi}^{\otimes 2}_{AA'BB'}(U_{AB}\otimes\chi_{B'})^\dag$ with $\|\chi\|_2=1$. Let $\chi_M$ be a $\chi$ that achieves the maximum.  Since the partial trace never decreases the fidelity, by tracing out systems other than $B'$, we get $F(I_{B'}/d,|\chi_M|^2)\geq \opt$. Let the recovery map that achieves the optimal fidelity be $R_\mathrm{opt}$ and let $\Theta_{V}:= V_{BB'}\dyad{\psi}_{AA'BB'}^{\otimes2} V_{BB'}^\dag$ for any bipartite unitary operator $V_{BB'}$, so that $\opt = F(\Omega_{\chi_M},\Theta_{R_\mathrm{opt}})$. Since there is a freedom of local unitary operation to the choice of $R_\mathrm{opt}$ and $\chi_M$, without loss of generality, we can assume that $\chi_M$ is positive semi-definite so that $\Tr \chi_M = \Tr |\chi_M|.$

From the following relation for arbitrary pure quantum states $\ket{\eta_1}$ and $\ket{\eta_2}$,
\begin{equation} \label{eqn:puredist}
\frac{1}{2}\|\dyad{\eta_1}-\dyad{\eta_2}\|_1 = \sqrt{1-|\bra{\eta_1}\ket{\eta_2}|^2},
\end{equation}
we have $\sqrt{1-\ME}=\min_W \|\Theta_W-\Omega_{\textsc{ME}}\|_1/2$, where $\Omega_\textsc{ME}:=\Omega_{{I_{B'}}/{\sqrt{d}}}$. Therefore $\sqrt{1-\ME}\leq\|\Theta_{R_\mathrm{opt}}-\Omega_\textsc{ME}\|_1/2$. Because of the triangular inequality, we have  $\|\Theta_{R_\mathrm{opt}}-\Omega_\textsc{ME}\|_1 \leq \|\Theta_{R_\mathrm{opt}}-\Omega_{\chi_M}\|_1+\|\Omega_{\chi_M}-\Omega_\textsc{ME}\|_1$. By Eq. (\ref{eqn:puredist}), we have $\|\Theta_{R_\mathrm{opt}}-\Omega_{\chi_M}\|_1/2 = \sqrt{1-\opt}$ and  $\|\Omega_{\chi_M}-\Omega_\textsc{ME}\|_1/2 = \sqrt{1-d^{-1}(\Tr \chi_M)^2}= \sqrt{1-F(I_{B'}/d,|\chi_M|^2)} \leq \sqrt{1-\opt}$. As a result, we have $\sqrt{1-\ME}\leq 2\sqrt{1-\opt}$ thus $1-\ME \leq 4(1-\opt).$

The trade-off relation between the data extraction fidelity $f_\mathrm{ext}$ and the posterior probe state fidelity $f_\mathrm{prob}$

\begin{equation}
    f_{\mathrm{ext}}+f_{\mathrm{prob}}\leq 1+\opt,
\end{equation}
directly follows from the following result.
\begin{theorem}
	For arbitrary bipartite quantum state $\rho_{AB}$ and pure states $\ket{\psi}_A$ and $\ket{\phi}_B$, let $F_A:=\bra{\psi}\rho_A\ket{\psi}$, $F_B:=\bra{\phi}\rho_B\ket{\phi}$ and $F_{AB}:=\bra{\psi}_A\bra{\phi}_B\rho_{AB}\ket{\psi}_A\ket{\phi}_B$. Then the following inequality holds
	\begin{equation}
	F_A+F_B\leq 1+F_{AB}.
	\end{equation}
\end{theorem}
\begin{proof}
It is enough to realize that $1-F_A-F_B+F_{AB}$ equals to $\Tr[\rho_{AB}(\psi^\perp_A\otimes\phi^\perp_B)]$ which is always non-negative. Here, $\psi^\perp_A:=\mathds
{1}_A-\psi_A$ is the projector onto the kernel of $\psi_A$ and similarly for $\phi^\perp_B$.
\end{proof}

Now, consider the black hole radiation problem of \textit{Hacking as entanglement recycling} section. When the probe state is a general mixed bipartite state $\Pi_{BB'}=\sum_i p_i \dyad{\phi_i}_{BB'}$ with $\ket{\phi_i}_{BB'}=\sum_k(I_B\otimes \chi_i)\ket{kk}_{BB'}$, the hacking fidelity is given as
\begin{equation}
p_\mathrm{hack}^{(R,\Pi)}=\sum_i p_i \frac{|\Tr[R (I_B \otimes \chi_i) U^o]|^2}{d_M^2D_B}.
\end{equation}
If the dimension of the Hilberst space of black hole state is large enough, then the PG strategy becomes nearly optimal thus, $p_\mathrm{hack}^{(R,\Pi)}$ reduces to $\sum_i p_i |\Tr[\chi_i \Tr_B |U^{o\dag}|]|^2/d_M^3.$ Moreover, as $D\to \infty$. $\Tr_B|U^{o\dag}|$ converges to $\|U^o\|_1 I_B'/D_B$(See Sec. \ref{app:asymp}), so we have
\begin{align}
\max_R p_\mathrm{hack}^{(R,\Pi)}&\approx\sum_i p_i \frac{|\Tr\chi_i|^2\|U^o\|_1^2}{D_B D^2}\nonumber\\
&=\ME\sum_i p_i \frac{|\Tr\chi_i|^2}{D_B}\nonumber\\
&\approx \opt f_\mathrm{prob}'.
\end{align}
Where $f_\mathrm{prob}'=\sum_i p_i {|\Tr\chi_i|^2}/{D_B}$ is the fidelity between $\Pi_{BB'}$ and a maximally entangled state. Therefore the hacking fidelity is asymptotically the product of the optimal hacking fidelity and $f_\mathrm{prob}'$.

\section{Duality with Hayden--Preskill protocols}
\label{app:dual}
Surprisingly, the seemingly harder problem of finding an optimal scrambler hacking strategy by Bob on Alice is equivalent to  that of a Hayden--Preskill-type protocol of Alice on Bob. In this setting we assume that, instead of $U$, its (computational-basis) transpose $(U^\top)^{ij}_{kl}=U^{kl}_{ij}$ is applied to systems $AB$. Now, Alice wants to extract information from Bob's system $B$. Similar to scrambler hacking, to model such information extraction, we assume that a maximally entangled state $\ket{\psi}_{BB'}$ is fed into $U^\top$. Alice also chooses a maximally entangled state $\ket{\psi}_{AA'}$ as a probe state. Systems $AB$ interact with $U^\top$ and Alice applies a $d_A^2$-dimensional unitary operator $W$ on $AA'$. Alice's goal is to prepare a maximally entangled state on systems $AB'$. The optimal fidelity between the actual and ideal states is $p_{\mathrm{HP}}^\mathrm{opt}=\max_W\bra{\psi}_{AB'}\Tr_{A'B}[\mathcal{W}\circ\mathcal{U}(\dyad{\psi}_{A'ABB'}^{\otimes 2})]\ket{\psi}_{AB'}$.
Here, $\mathcal{W}(\rho):=W_{AA'}\rho\, W_{AA'}^\dag\equiv W\rho\,W^\dag$ and $\mathcal{U}(\rho):= U^\top_{AB}\rho\,U_{AB}^*\equiv U^\top\rho\,U^*$. This expression can also be simplified in terms of the \texttt{SWAP} operator $F$ to
\begin{equation} \label{HPopt}
    p_{\mathrm{HP}}^\mathrm{opt}=\max_W\|\Tr_B[U^o\,W^\top\! F]\|_2^2/(d_A d_B^2).
\end{equation}
It follows that the optimal hacking strategy of 
\begin{equation} \label{opthack}
    p_{\mathrm{hack}}^{\mathrm{opt}}=\max_R\|\Tr_B[U^o R]\|_2^2/d_A^3,
\end{equation}
and the optimal strategy to $p_{\mathrm{HP}}$ are related by $R=W^\top\!F$. Therefore, finding an optimal strategy for scrambler hacking is formally equivalent to finding an optimal strategy for the Hayden--Preskill protocol, in the sense that if one problem is solvable for an arbitrary $U$, then so is the other. It follows that $p_{\mathrm{HP}}^\mathrm{opt}=\opt/\kappa^{2}$ and $p_{\mathrm{HP}}^\mathrm{opt}<1$ when $d_B>d_A$.

\section{Numerical maximization of $p^\mathrm{opt}_\mathrm{hack}$}
\label{app:numer}

{Given a scrambler} described by $U$, it is possible to derive an iterative numerical scheme to obtain the optimal probe state ($\chi_\mathrm{opt}$) that achieves the optimal hacking fidelity $p^\mathrm{opt}_\mathrm{hack}$. Rather than directly solving the numerical problem in~(5) of the main text, we can instead start with $f_\chi=\|({I_L}\otimes\chi)U^o\|_1$, which is the objective function involving the square root of the rightmost side in~(4), and perform a variation with respect to $\chi$. Furthermore, the constraint $\|\chi\|_2=1$ invites the following parametrization $\chi=Z/\|Z\|_2$, such that 
\begin{equation}
\updelta\chi=\dfrac{\updelta Z}{\|Z\|_2}-\dfrac{Z}{2\|Z\|^3_2}{\Tr}[\updelta ZZ^\dag+Z\updelta Z^\dag]\,.
\end{equation}
Upon denoting $M=({I_L}\otimes \chi)U^o$, we consequently have
\begin{align}
\updelta f_\chi=&\,\dfrac{1}{2}{\Tr}\left[\Tr_{B}[|M^\dag|^{-1}MU^{o\dag}]\dfrac{\updelta Z^\dag}{\|Z\|_2}\right]+\mathrm{c.c.}\nonumber\\
&\,-\dfrac{1}{2}\Tr|M^\dag|\dfrac{{\Tr}[\updelta ZZ^\dag+Z\updelta Z^\dag]}{\|Z\|^2_2}\,,
\end{align}
which leads to the operator gradient
\begin{equation}
\dfrac{\updelta f_\chi}{\updelta Z^\dag}=\dfrac{1}{2\|Z\|_2}\left(\Tr_{B}[|M^\dag|^{-1}MU^{o\dag}]-\Tr|M^\dag|\dfrac{Z}{\|Z\|_2}\right)
\end{equation}
with respect to $Z^\dag$. Setting it to zero would then gives the extremal equation
\begin{equation}
\chi=\dfrac{\Tr_{B}[|M^\dag|^{-1}MU^{o\dag}]}{\|\Tr_{B}[|M^\dag|^{-1}MU^{o\dag}]\|_2}\,,
\label{eqn:ext_eqn}
\end{equation}
which may alternatively be gotten from reasoning with the Cauchy-Schwarz inequality. As $\|({I_L}\otimes\chi)U^o\|_1$ is concave in $\rho_{B'}=\chi^\dag\chi$, one can generally expect a convex solution set of $\rho_{B'}$'s that solve \eqref{eqn:ext_eqn}, all of which give the unique maximal fidelity $p^\mathrm{opt}_\mathrm{hack}$.

In other words, $p^\mathrm{opt}_\mathrm{hack}$ is achieved when a solution $\chi=\chi_\mathrm{opt}$ for Eq.~\eqref{eqn:ext_eqn} is obtained. While there are no known closed-form expressions for this solution, we can nevertheless find explicit analytical forms for certain special cases. The most immediate one happens to be the limiting case $d_B\rightarrow\infty$, whence we have $\chi_\mathrm{opt}\rightarrow I_{B}/\sqrt{d_{B}}$, since in this limit, {$\Tr_BO\rightarrow I_{B}\Tr O/d_B$ for any bipartite operator $O$} of systems $BB'$. For finite $d_B$, we may still have an estimate for $\chi_\mathrm{opt}\approx\widetilde{\chi}$. A straightforward way to do this is to simply iterate the extremal equation \eqref{eqn:ext_eqn} once by substituting $I_{B}/\sqrt{d_{B}}$ for $\chi$ on the right-hand side. This gives us $\widetilde{\chi}=\Tr_B|U^{o\dag}|/\|\Tr_{B}|U^{o\dag}|\|_2$, which is in practice very close to $\chi_\mathrm{opt}$.

{In practice, iterating Eq.~\eqref{eqn:ext_eqn} usually results in good convergence to $\chi_\mathrm{opt}$. For the pedantic, we may additionally} adopt the steepest-ascent methodology and require that $\updelta f_\chi=\Tr\left[(\updelta f_\chi/\updelta Z^\dag)\updelta Z^\dag+\updelta Z(\updelta f_\chi/\updelta Z)\right]\geq0$. This amounts to defining the increment $\updelta Z:=\epsilon\,\updelta f_\chi/\updelta Z^\dag$ for some small real $\epsilon>0$ that functions as a fixed iteration step size. This allows us to state the iterative equations
\begin{align}
Z_{k+1}=&\,\left(1-\dfrac{\epsilon}{2}\dfrac{\Tr|M_k^\dag|}{\|Z_k\|_2}\right)Z_k+\dfrac{\epsilon}{2}\Tr_{B}[|M_k^\dag|^{-1}M_kU^{o\dag}]\,,\nonumber\\
\chi_{k+1}=&\,\dfrac{Z_{k+1}}{\|Z_{k+1}\|_2}
\end{align}
that can be used to converge $\chi_k$ to $\chi_\mathrm{opt}$ starting with $Z_1=I_B$, where a factor of $\|Z_k\|_2$ has been neglected for a suitably chosen magnitude of $\epsilon$. As $\updelta f_\chi=2\epsilon\Tr[|\updelta f_\chi/\updelta Z|^2]>0$ by construction, convergence is guaranteed as long as $\epsilon$ is sufficiently small. Operationally, one can afford to choose a reasonably large $\epsilon$ to increase the convergence rate.

\section{{Optimal hacking} of two-qubit quantum scramblers}
\label{app:two-qubit}

The case where {$d_A=d_B=d_K=d_L=2$} presents the unique situation in which one can confirm, indeed, that $\chi_\mathrm{opt}=I_B/\sqrt{d_B}$. To this end, we proceed to construct the exact expression of $|U^{o\dag}|$. Since $UU^\dag=I$, in terms of the product computational basis $\bra{jk}U\ket{lm}=U^{jk}_{lm}$, the basic relation
\begin{equation}
\sum^1_{l,m=0}U^{j_1k_1}_{lm}U^{j_2k_2*}_{lm}=\delta_{j_1,j_2}\delta_{k_1,k_2}
\label{eqn:unitarity}
\end{equation}
shall be immensely useful in the subsequent discussion.

Using Eq.~\eqref{eqn:unitarity}, we first note that the product
\begin{align}
U^oU^{o\dag}=&\sum_{l,m,m'}\!\Big[\ket{0,m}(U^{00}_{lm}U^{00*}_{lm'}+U^{10}_{lm}U^{10*}_{lm'})\bra{0,m'}\nonumber\\[-1ex]
&\,\qquad+\ket{1,m}(U^{01}_{lm}U^{01*}_{lm'}+U^{11}_{lm}U^{11*}_{lm'})\bra{1,m'}\nonumber\\
&\,\qquad+\ket{0,m}(U^{00}_{lm}U^{01*}_{lm'}+U^{10}_{lm}U^{11*}_{lm'})\bra{1,m'}\nonumber\\
&\,\qquad+\ket{1,m}(U^{01}_{lm}U^{00*}_{lm'}+U^{11}_{lm}U^{10*}_{lm'})\bra{0,m'}\Big]\nonumber\\
\,\widehat{=}&\,\begin{pmatrix}
\mathbf{A} & \mathbf{B}^\dag\\
\mathbf{B} & \mathbf{A}^{-1}\mathrm{Det}\mathbf{A}
\end{pmatrix}
\end{align}
may be characterized, in the product computational basis, by only two $2\times2$ matrices in a highly specific manner, where $\mathbf{B}$ is traceless. Such a structure is absent in higher dimensions. For convenience, we may rewrite
\begin{equation}
U^oU^{o\dag}\,\widehat{=}\,\bm{1}+\begin{pmatrix}
\bm{a\cdot\sigma} & \bm{b}^*\bm{\cdot\sigma}\\
\bm{b\cdot\sigma} & -\bm{a\cdot\sigma}
\end{pmatrix}
\end{equation}
in terms of dot products ($\bm{v\cdot w}=\bm{v}^\top\bm{w}$) of the vectorial parameters $\bm{a}$ and $\bm{b}$ with the standard vector of Pauli operators $\bm{\sigma}=(\sigma_x,\sigma_y,\sigma_z)^\top$ to separate the matrix representation of $U^oU^{o\dag}$ into the identity and another $4\times4$ traceless matrix, where $\bm{a}$ is real and $\bm{b}$ complex.

With the identity $(\bm{a\cdot\sigma})(\bm{a}'\bm{\cdot\sigma})=\bm{a\cdot a}'\bm{1}+\I\,\bm{a\times a}'\bm{\cdot\sigma}$, it is a straightforward matter to verify that $|U^{o\dag}|$ has the same matrix-representation structure, where all its parameters satisfy the following conditions:
\begin{align}
|U^{o\dag}|\,\widehat{=}&\,c'\bm{1}+\begin{pmatrix}
\bm{a}'\bm{\cdot\sigma} & \bm{b}'^*\bm{\cdot\sigma}\\
\bm{b}'\bm{\cdot\sigma} & -\bm{a}'\bm{\cdot\sigma}
\end{pmatrix}\,,\nonumber\\
1=&\,c'^2+|\bm{a}'|^2+|\RE{\bm{b}'}|^2+|\IM{\bm{b}'}|^2\,,\nonumber\\
\bm{a}=&\,2\,c'\bm{a}'-2\,\RE{\bm{b}'}\bm{\times}\IM{\bm{b}'}\,,\nonumber\\
\RE{\bm{b}}=&\,2\,c'\RE{\bm{b}'}-2\,\IM{\bm{b}'}\bm{\times}\bm{a}'\,,\nonumber\\
\IM{\bm{b}}=&\,2\,c'\IM{\bm{b}'}-2\,\bm{a}'\bm{\times}\RE{\bm{b}'}\,.
\end{align}
Here, $\RE{\cdot}$ and $\IM{\cdot}$ respectively denote the real and imaginary parts of the argument. Evidently, in this fortuitously easy yet general two-qubit scenario, we find that $\Tr_B|U^{o\dag}|=2\,c'I_B$, such that $\widetilde{\chi}=\Tr_B|U^{o\dag}|/\|\Tr_B|U^{o\dag}|\|_2=I_B/\sqrt{d_B}=\chi_\mathrm{opt}$.

\section{Asymptotic formulas for $p^\mathrm{opt}_\mathrm{hack}$} \label{app:asymp}

{For the problem of quantum-information extraction involving unitary scrambling dynamics described by a $d_Ad_B$-dimensional unitary operator $U$, we may consider a general a very general setting where the bipartite output dimensions of $U$ are respectively $d_K$ and $d_L$ for systems $A$ and $B$, such that clearly $d_Ad_B=d_Kd_L$ owing to unitarity. In the asymptotic limit $d_A,d_B\rightarrow\infty$, according to the discussions in Sec.~\ref{app:numer}}, the corresponding optimal scrambler hacking fidelity takes the form {$p^\mathrm{opt}_\mathrm{hack}\rightarrow\|U^o\|^2_1/(d_A^2d_Bd_K)$}. The analytical form of its average value then necessitates {the calculation of} the average term $\overline{\|U^o\|^2_1}$ over all random $U$'s distributed according to {some specific distribution measure, which we fix to be the Haar measure}. We emphasize that since $U^o$ is represented by a $d_Bd_L\times d_Ad_K$ matrix that is obtained from just a sequence of index swapping operations, such a rectangular matrix still retains the statistical properties of a Haar unitary matrix elements, namely $\overline{U^{ojk}_{\hphantom{o}lm}}=0$ and $\overline{|U^{ojk}_{\hphantom{o}lm}|^2}=1/(d_Ad_B)$~\cite{mehta2004random}. If we additionally suppose that {$\kappa\equiv\sqrt{d_Bd_L/(d_Ad_K)}=d_B/d_K\geq1$}, then in the {dimensional} asymptotic limit, {the random Haar unitary ensemble has elements that are so weakly correlated that they are approximately independently and identically distributed~\cite{Jonnadula2020entanglement,Mingo2021asymptotic}. Each eigenvalue $(\sigma_j)$ of the positive operator $\kappa^{-1} U^{o\dag}U^o$ shall then independently follow} the Mar{\v c}enko--Pastur distribution~\cite{marcenko1967eigenvalues}:
\begin{equation}
\sigma_j\sim\dfrac{1}{2\pi}\dfrac{\sqrt{(\lambda_+-x)(x-\lambda_-)}}{\lambda x}\,,\quad \lambda_{\pm}=(1\pm\sqrt{\lambda})^2\,,
\label{eqn:MPdist}
\end{equation}
{which is specified by the distribution's characteristic variable} $\lambda=\kappa^{-2}$. With these,
{\begin{align}
\overline{\|U^o\|^2_1}=&\,\kappa\left(\sum^{d_Ad_K}_{j=1}\overline{\sigma_j}+\sum_{j\neq k}\overline{\sqrt{\sigma_j}}\,\overline{\sqrt{\sigma_k}}\right)\nonumber\\
=&\,\kappa\left[d_Ad_K+(d_Ad_K-1)d_Ad_K\mathcal{I}^2_\kappa\right]\,,
\end{align}}where $\overline{\,\,\cdot\,\,\vphantom{M}}$ now translates to an average with respect to the distribution in~\eqref{eqn:MPdist}{, and we have used the fact that $\overline{\sigma_j}=1$}. The quantity $\mathcal{I}_\kappa=\overline{x^{1/2}}$ refers to the half-moment of this distribution. For completeness, we evaluate the $m$th moment: 
\begin{align}
&\,\overline{x^m}=\int^{\lambda_+}_{\lambda_-}\,\dfrac{\D x}{2\pi\lambda}\,x^{m-1}\sqrt{(\lambda_+-x)(x-\lambda_-)}\nonumber\\
=&\,\dfrac{2}{\pi}(1+\lambda)^{m-1}\int^{1}_{-1}\,\D t\left(1+\dfrac{2\sqrt{\lambda}}{1+\lambda}\,t\right)^{m-1}\sqrt{1-t^2}\nonumber\\
=&\,\left(1+\dfrac{1}{\kappa^2}\right)^{m-1}{}_2\mathrm{F}_1\left(\dfrac{1-m}{2},1-\frac{m}{2};2;\left(\dfrac{2/\kappa}{1+1/\kappa^2}\right)^2\right)\,.
\label{eqn:mth}
\end{align}
The variable substitution $x=(\lambda_++\lambda_-)/2+(\lambda_+-\lambda_-)t/2$ has been introduced after the second equality in \eqref{eqn:mth}. We emphasize that the last equality in \eqref{eqn:mth} is valid for \emph{any real} $m$ so long as the previous $t$~integral converges. {Hence, $\overline{x^m}$ is proportional to a \emph{hypergeometric function} ${}_2\mathrm{F}_1(\,\cdot\,,\,\cdot\,;\,\cdot\,;\,\cdot\,)$.} Upon using the identity~\cite{SpFuncBk}
\begin{equation}
{}_2\mathrm{F}_1(2a,2a+1-\gamma;\gamma;z)=\dfrac{{}_2\mathrm{F}_1\left(a,a+\dfrac{1}{2};\gamma;\dfrac{4z}{(1+z)^2}\right)}{(1+z)^{2a}}\,,
\end{equation}
we get $\overline{x^m}={}_2\mathrm{F}_1(1-m,-m;2;{\kappa^{-2}})$. {That $\overline{x}=1$ follows immediately from a direct evaluation of the hypergeometric function.} Thereafter, the substitution $m=1/2$ nabs us the final answer $\mathcal{I}_\kappa={}_2\mathrm{F}_1\left({2^{-1},-2^{-1};2;\kappa^{-2}}\right)$, so that
{\begin{equation} \overline{p^{\mathrm{opt}}_\mathrm{hack}}\approx\mathcal{I}_\kappa^2+\dfrac{1}{d_Ad_K}(1-\mathcal{I}_\kappa^2)\quad\mathrm{for}\,\,\kappa\geq1\,. 
\end{equation}}Moreover, we may simplify this expression further by considering a moderately large $\kappa$, for which the hypergeometric function has the simple second-order approximation $\mathcal{I}_\kappa\approx1-1/(8\kappa^2)$. This simplification works amazingly well even for $\kappa=1$---$\mathcal{I}_1=8/(3\pi)\approx0.875$---such that one might as well use this approximation for any $\kappa$.

Now, if $\kappa<1$, one can go through a similar line of argument and arrive at {$\overline{\|U^o\|^2_1}=\kappa^{-1}\left[d_Bd_L+(d_Bd_L-1)d_Bd_L\mathcal{I}^2_{1/\kappa}\right]$}, in which case, we get 
{\begin{equation} \overline{p^{\mathrm{opt}}_\mathrm{hack}}\approx\kappa^2\,\mathcal{I}_{1/\kappa}^2+\dfrac{1}{d_Ad_K}(1-\mathcal{I}_{1/\kappa}^2)\quad\mathrm{for}\,\,\kappa<1\,,
\end{equation}
which} tells us that the asymptotic optimal scrambler hacking fidelity is going to be smaller than that when {$\kappa\geq1$.}

\end{acknowledgments}

\bibliographystyle{unsrtnat}
\bibliography{main}

\begin{thebibliography}{44}
\bibitem[1()]{1}
Although $R$ is $d_B^2\times d_B^2$, it only acts on a $d_A^2$-dimensional
  subspace $\Im \{(I_B \otimes \chi )U^o\}$ to the right and $(\mathrm {Ker }
  \{U^o\})^\perp $ to the left, so that we may treat $R$ either as a
  rank-$d_A^2$ partial unitary matrix or a $d_A^2 \times d_B^2$ coisometry
  without losing generality.
\end{thebibliography}
\end{document}